\documentclass[12pt]{article} 
\usepackage[sectionbib]{natbib}
\usepackage{array,epsfig,fancyheadings,rotating}
\usepackage[]{hyperref}  
\usepackage{sectsty, secdot}
\sectionfont{\fontsize{12}{14pt plus.8pt minus .6pt}\selectfont}
\renewcommand{\theequation}{\thesection\arabic{equation}}
\subsectionfont{\fontsize{12}{14pt plus.8pt minus .6pt}\selectfont}

\usepackage{amsmath}
\usepackage{amssymb}
\usepackage{amsfonts}
\usepackage{multirow}
\usepackage{amsthm}
\usepackage{graphicx}
\usepackage{color}

\def\modq {\pmod{q}}
\newcommand{\reva}[1]{{\color{red} #1}}
\renewcommand{\reva}[1]{{#1}}
\newcommand{\revb}[1]{{\color{blue} #1}}
\renewcommand{\revb}[1]{{#1}}

\newtheorem{theorem}{Theorem}
\newtheorem{lemma}{Lemma}
\newtheorem{corollary}{Corollary}

\theoremstyle{definition}
\newtheorem{definition}{Definition}
\newtheorem{example}{Example}
\newtheorem{remark}{Remark}
\begin{document}


\renewcommand{\baselinestretch}{2}

\markright{ \hbox{\footnotesize\rm Statistica Sinica
}\hfill\\[-13pt]
\hbox{\footnotesize\rm
}\hfill }

\markboth{\hfill{\footnotesize\rm LIN WANG AND HONGQUAN XU} \hfill}
{\hfill {\footnotesize\rm A CLASS OF MULTILEVEL NONREGULAR DESIGNS} \hfill}

\renewcommand{\thefootnote}{}
$\ $\par


\fontsize{12}{14pt plus.8pt minus .6pt}\selectfont \vspace{0.8pc}
\centerline{\large\bf A CLASS OF MULTILEVEL NONREGULAR DESIGNS}
\vspace{2pt} \centerline{\large\bf FOR STUDYING QUANTITATIVE FACTORS}
\vspace{.4cm} \centerline{Lin Wang$^1$, \ Hongquan Xu$^2$} \vspace{.4cm}
\centerline{\it $^1$The George Washington University}
\centerline{\it and $^2$University of California, Los Angeles}
\vspace{.55cm} \fontsize{9}{11.5pt plus.8pt minus
.6pt}\selectfont


\begin{quotation}
\noindent {\it Abstract:}
Fractional factorial designs are widely used for designing screening experiments. Nonregular fractional factorial designs can have better properties than regular designs, but their construction is challenging. Current research on the construction of nonregular designs focuses on two-level designs. We provide a novel class of multilevel nonregular designs by permuting levels of regular designs. We develop a theory illustrating how levels can be permuted without computer search and accordingly propose a sequential  method for constructing nonregular designs. Compared to regular designs, these nonregular designs can provide more accurate estimations on factorial effects and more efficient screening for experiments with quantitative factors. We further explore the space-filling property of the obtained designs and demonstrate their superiority.

\vspace{9pt}
\par
\end{quotation}\par

\def\thefigure{\arabic{figure}}
\def\thetable{\arabic{table}}

\renewcommand{\theequation}{\thesection.\arabic{equation}}

\fontsize{12}{14pt plus.8pt minus .6pt}\selectfont

\setcounter{section}{0} 
\setcounter{equation}{0} 

\lhead[\footnotesize\thepage\fancyplain{}\leftmark]{}\rhead[]{\fancyplain{}\rightmark\footnotesize\thepage}

\section {Introduction}
Screening experiments are commonly designed to investigate the controlled factors and identify important ones.
Fractional factorial designs are very suitable for screening experiments because they allow the investigation of many factors simultaneously with a small number of runs.
These designs are classified into two broad types: regular designs and nonregular designs.  Designs that can be constructed
through defining relations among factors are called regular designs, {while all other designs are nonregular.} There are many more nonregular designs than regular designs. Good nonregular designs can either fill the gaps between regular designs in terms of various run sizes or provide better estimation for factorial effects.

The construction of good nonregular designs is important and challenging. {Constructions for two-level nonregular designs include
\cite{plackett1946design}, \cite{deng2002design},   \cite{XuDeng2005}, \cite{fang2007effective},  \cite{phoa2009quarter}, among others.}
While numerous constructions are available for two-level designs, these designs are not able to provide information on quadratic or high-order factorial effects.
Multilevel designs with three or more levels are in high demand in many scientific and engineering fields \reva{such as many recent studies on drug combination experiments \citep{ding2013use, jaynes2013, silva2016output, clemens2019artificial}} because these designs provide the capability of studying complex factorial effects and interactions. They are also flexible on designing the number of levels for factors, without the strict restriction with \reva{Latin hypercube designs (LHDs)} that the number of levels has to be the same as the run size.
Nevertheless, constructions for multilevel nonregular designs rarely exist \citep{xu2009recent}.
This is because the number of multilevel nonregular designs is huge so that providing an efficient algorithm for searching the design space is super challenging. A systematic construction also seems impossible without a unified mathematical description.

This paper provides a class of multilevel nonregular designs by manipulating nonlinear level permutations on regular designs. While linear level permutations have been studied by \cite{cheng2001factor}, \cite{xu2004optimal}, and \cite{ye2007optimal} for three-level designs, and by \cite{tang2014permuting} to improve properties of regular designs, nonlinear level permutations have not been studied.
Note that linearly permuted regular designs can be still considered as regular because they are just cosets of regular designs and share the same defining relationship.
We consider a nonlinear level permutation via the Williams transformation,
which was first used by \cite{williams1949experimental} to construct balanced Latin square designs, followed by \cite{butler2001optimal} and \cite{wang2018optimal} to construct orthogonal or maximin LHDs. Our purpose is different from theirs.
We provide a class of nonregular designs by manipulating nonlinear level permutations on regular designs via the Williams transformation and develop a general theory on the obtained designs. Using the theory, we propose a sequential construction method that efficiently constructs good designs in terms of the minimum $\beta$-aberration criterion, a criterion that assesses multilevel designs.
We further explore the space-filling property of the obtained designs and demonstrate their superiority.

The paper is organized as follows. Section 2 introduces the minimum $\beta$-aberration criterion and generates a class of nonregular designs via the Williams transformation. Section 3 presents our main theoretical results. Based on the theory, in Section 4 we propose a sequential construction method and compare the constructed designs with available designs.
In Section 5, we consider the application of the constructed designs. Section 6 concludes the paper and all proofs are deferred to the Appendix.

\section{Notation, Background and Definitions}
Let $Z_q=\{0, \ldots, q-1\}$. A $q$-level design  $D=(x_{ij})$ with $N$ runs and $n$ factors is an $N \times n$ matrix over $Z_q$ where each column corresponds to a factor.
\reva{Let $p_{0}(x)\equiv1$ and $p_{j}(x)$ for $j=1,\ldots,q-1$ be an orthonormal polynomial of order $j$ defined on $Z_{q}$ satisfying
$$
\sum_{x=0}^{q-1}p_i(x)p_j(x)=\left\{
                               \begin{array}{ll}
                                 0, & \hbox{$i\neq j$;} \\
                                 q, & \hbox{$i=j$.}
                               \end{array}
                             \right.
$$
The set $\{p_0(x), p_1(x), \ldots, p_{q-1}(x)\}$ is called an orthonormal polynomial basis.
}

Multilevel designs are often used for studying quantitative factors by fitting response surface models such as polynomial models.
A commonly accepted principle for polynomial models is that effects of a lower polynomial order are more important than effects of a higher polynomial order, while effects of the same polynomial order are regarded as equally important.
Based on this principle, \cite{cheng2004geometric} proposed the {\em minimum $\beta$-aberration} criterion for selecting multilevel designs.
\reva{For a $q$-level design $D=(x_{ij})$ with $N$ runs and $n$ factors}, define
\begin{equation}\label{eqbeta}
\beta_{k}(D)=N^{-2}\sum_{\|u\|_{1}=k}\left|\sum_{i=1}^{N}\prod_{j=1}^{n}p_{u_{j}}(x_{ij})\right|^2 \mbox{ for }k=1,\ldots,K,
\end{equation}
where {$u=(u_1, \ldots, u_n) \in Z_q^n$,} $\|u\|_1=u_1+\cdots+u_n$ and $K=n(q-1)$.
The vector $(\beta_1(D),\ldots,\beta_K(D))$ is called the $\beta$-wordlength pattern of $D$ and each $\beta_k$ measures the overall aliasing between $j$th- and $(k-j)$th-order polynomial terms for all $j$ with $0\leq j\leq k$.
\reva{The minimum $\beta$-aberration criterion is to sequentially minimize $\beta_k$ for $k=1,2,\ldots,K$. Because linear and second-order terms are more important than higher-order terms, the sequential minimization of $\beta_1,\ldots,\beta_4$ would be adequate for choosing designs in practice. \cite{tang2014permuting} and \cite{lin2017minimum} provided statistical justification and additional insights regarding  minimum $\beta$-aberration designs.

The minimum $\beta$-aberration criterion is an extension of the minimum $G_2$-aberration criterion \citep{tang1999minimum} for two-level designs, but is different from the generalized minimum aberration criterion \citep{xu2001generalized} for multi-level designs with qualitative factors.}

{For $x\in Z_q$}, the Williams transformation is defined by
\begin{equation}\label{eq6}
W(x)=\left\{
             \begin{array}{ll}
               2x, & \hbox{for $0\leq x<q/2$;} \\
               2(q-x)-1, & \hbox{for $q/2\leq x<q$.}
             \end{array}
           \right.
\end{equation}
The Williams transformation is a permutation of \revb{$Z_q$.}
For a design $D=(x_{ij})$, let $W(D)=(W(x_{ij}))$.
The following example shows that we can get better designs from the Williams transformation.

\begin{example}\label{25run}\rm
Consider a 5-level regular design $D$ with three columns $x_1$, $x_2$ and $x_3=x_1+x_2\pmod{5}$.
By \eqref{eqbeta}, $\beta_1(D)=\beta_2(D)=0$, $\beta_3(D)=0.125$, and $\beta_4(D)=0.525$.
For each $b=0,\ldots,4$, we obtain two designs via linear permutations and the Williams transformation, namely, $D_b$ with columns $x_1$, $x_2$ and $x_3=x_1+x_2+b\pmod{5}$ and $E_b=W(D_b)$. It can be verified that all $D_b$'s and $E_b$'s have $\beta_1=\beta_2=0$.
Table \ref{tab0} shows their $\beta_3$ and $\beta_4$. The best design from $D_b$'s is $D_3$ with $\beta_3=0$ and $\beta_4=0.686$, while the best design from $E_b$'s is $E_4$ with $\beta_3=0$ and $\beta_4=0.027$.
Design $E_4$ performs much better than $D_3$ under the minimum $\beta$-aberration criterion, although they are both better than the original design $D$.
\end{example}
\begin{table}[t!]
\caption{The $\beta$-wordlength pattern of $D_b$ and $E_b$ in Example \ref{25run}. \label{tab0}}
\vskip .2cm
\centerline{\tabcolsep=5truept\begin{tabular}{|rrrrrrr|} \hline
$b$&& $\beta_3(D_{b})$ &  $\beta_4(D_{b})$  &&   $\beta_3(E_{b})$  & $\beta_4(E_{b})$\\\hline
0 && 0.125 & 0.525 && 0.442 & 0.004 \\
1 && 0.125 & 0.525 && 0.168 & 0.021 \\
2 && 0.125 & 0.096 && 0.168 & 0.021 \\
3 && 0.000 & 0.686 && 0.442 & 0.004 \\
4 && 0.125 & 0.096 && 0.000 & 0.027 \\
\hline \end{tabular}}
\end{table}

\begin{remark}\rm
In the computation of $\beta_k$ defined in \eqref{eqbeta}, \reva{the orthonormal polynomials for a 5-level factor} are $p_0(x)=1$, $p_1(x)=(x-2)/\sqrt{2}$,  $p_2(x)=\sqrt{10/7}\{p_1(x)^2-1\}$,
$p_3(x)=\{10p_1(x)^3-17p_1(x)\}/6$, and $p_4(x)=\{70p_1(x)^4-155p_1(x)^2+36\}/\sqrt{14}$.
\end{remark}

Example 1 shows that from a regular design, we can obtain a series of nonregular designs via linear permutations and the Williams transformation.
This series of designs can provide better designs than the original regular design and linearly permuted designs.

Generally, for a prime number $q$, a regular $q^{n-m}$ \reva{design $D$} has $n-m$ independent columns, denoted as $x_{1},\ldots,x_{n-m}$, and $m$ dependent columns, denoted as $x_{n-m+1},\ldots,x_{n}$, which can be specified by $m$ generators as
\begin{equation}\label{eq31}
x_{n-m+i}=c_{i1}x_{1}+\cdots+c_{i(n-m)}x_{n-m} \pmod{q}, \mbox{ for } i=1,\ldots,m,
\end{equation}
where each vector $(c_{i1},\ldots,c_{i(n-m)})$ is a generator whose entries are constants in $Z_q$.
For each regular $q^{n-m}$ \reva{design  $D$} and $b=(b_1,\ldots,b_m)\in Z_q^m$,  let
\begin{equation}\label{eq2}
D_b=(x_1,\ldots,x_{n-m},x_{n-m+1}+b_1,\ldots,x_n+b_m)\modq,
\end{equation}
and
\begin{equation}\label{eqeb}
E_b=W(D_b).
\end{equation}
Note that we only consider permutations for dependent columns in \eqref{eq2} because linearly permuting one or more independent columns is equivalent to linearly permuting some dependent columns, which can be seen from \eqref{eq31}.
\revb{Throughout the paper, all additions between columns of a design are subject to the modulus $q$, the number of levels of the design, as in \eqref{eq31} and \eqref{eq2}. We omit the notation $\modq$ for such operations when no confusion is introduced.}
From each regular $q^{n-m}$ design $D$, we can derive $q^m$ $D_b$'s and $q^m$ $E_b$'s.
To find the best design,  one can search over all possible permutations $b\in Z_q^m$, which is, however, cumbersome and even infeasible in many cases. In the next section we develop a theory to determine the best $E_b$ without computer search.


{For $q=3$, the two classes of designs, $D_b$'s and $E_b$'s, always have the same $\beta$-wordlength patterns because they are geometrically isomorphic \citep{cheng2004geometric}.
However, with more than three levels, their performances are pretty different under the minimum $\beta$-aberration criterion. \cite{tang2014permuting} studied the class of $D_b$'s.  As we have seen in Example \ref{25run}, the class of $E_b$'s can provide many better designs than the class of $D_b$'s.}

\section{Theoretical Results}\label{sth}
We study  properties of $E_b$ in this section. It is well known that a regular design $D$ is an orthogonal array of strength $t\geq2$.
An orthogonal array is a design in which all $q^t$ level combinations appear equally often in every submatrix formed by $t$ columns. The $t$ is called the strength of the orthogonal array, which is often omitted when $t=2$.
Because the Williams transformation is a permutation of $\{0,\ldots,q-1\}$, if  $D=(x_{ij})$ is a $q$-level orthogonal array, then $W(D)=(W(x_{ij}))$ is still an orthogonal array.
The following result  is from \cite{tang2014permuting}.

\begin{lemma}\label{lem4}
 For an orthogonal array of strength $t$, $\beta_k=0$ for $k=1,\ldots,t$.
\end{lemma}

From the construction in \eqref{eqeb},  $E_b$ is an orthogonal array of the same strength as $D$ and $D_b$.
While we use designs of strength 2 in practice,
Lemma \ref{lem4} guarantees $\beta_1(E_b)=\beta_2(E_b)=0$ so that linear terms are not aliased with the intercept, nor with each other.
Then we want to minimize $\beta_3(E_b)$ in order to minimize the aliasing between linear and second-order {terms. } The following theorem gives a permutation $b$ theoretically to ensure $\beta_3(E_b)=0$ so that no aliasing exists between any linear and second-order terms.

\begin{theorem}\label{th1}
For an odd prime $q$, let
\begin{equation}\label{eqgamma}
\gamma=W^{-1}((q-1)/2)=\left\{
         \begin{array}{ll}
           (q-1)/4, & \hbox{if $q=1 \pmod{4}$;} \\
           (3q-1)/4, & \hbox{if $q=3\pmod{4}$.}
         \end{array}
       \right.
\end{equation}
Let $D$ be a regular $q^{n-m}$ design generated by \eqref{eq31}, and $E_b$ be defined by \eqref{eqeb}. Then  $\beta_3(E_{b^*})=0$ with  $b^*=(b^*_1, \ldots, b^*_m)$, where
\begin{equation}\label{eq1}
b_i^* = \left(1-\sum_{j=1}^{n-m}c_{ij} \right) \gamma   ~~~(i=1,\ldots,m).
\end{equation}
\end{theorem}

\begin{example}\label{exa4}\rm
Consider a $7^{3-1}$ design $D$ with $x_3=x_1+x_2$. Then $\gamma=(3\times7-1)/4=5$, and equation \eqref{eq1} gives $b_1^*=2$. It can be verified that $\beta_3(E_{2})=0$ and $\beta_4(E_{2})=0.003$. Consider another $7^{3-1}$ design $D$ with $x_3=2x_1+2x_2$. Then $\gamma=5$, and equation \eqref{eq1} gives $b_1^*=6$. It can be verified that $\beta_3(E_{6})=0$ and $\beta_4(E_{6})=0.0196$.
\end{example}



Theorem \ref{th1} states that given a regular design $D$, we can always find an $E_{b^*}$ such that $\beta_3(E_{b^*})=0$. In the following, we give a sufficient condition for the $E_{b^*}$ to be the unique design with $\beta_3=0$ among all possible $q^m$ $E_b$'s.

\begin{definition}\label{def1}
Let $D$ be a regular $q^{n-m}$ design. If there exist $n-m$ independent columns of $D$, $z_{1},\ldots,z_{n-m}$, and a series of $s+1$ sets of columns, $T_0\subset\cdots\subset T_s$, such that $T_0=\{z_{1},\ldots,z_{n-m}\}$,
\begin{equation}\label{eq7}
T_{k+1}=T_{k}\cup\{ {w\in D: w=c_1w_1+c_2w_2\modq, w_1,w_2\in T_k}, c_1,c_2\in Z_q\}
\end{equation}
for $k=0,\ldots,s-1$,
and $T_s=D$, then $D$ is called recursive. Furthermore, if either $c_1$ or $c_2$ is restricted to $1$ or $-1$ in  \eqref{eq7} for all $k$, then $D$ is called ordinary-recursive; if both $c_1$ and $c_2$ are resticted to $1$ or $-1$  in \eqref{eq7} for all $k$, then $D$ is called simple-recursive.
\end{definition}


\begin{example}\label{exa2}\rm
Consider the $7^{3-1}$ design $D$ defined by $x_3=2 x_1+2 x_2$ in Example \ref{exa4}. 
Clearly, $D$ is recursive. Because $-1=6\pmod 7$, we have $2 x_1+2 x_2+6x_3=0$,  $x_1+x_2+3x_3=0$ and $x_2=-x_1+4x_3$.
Then $D$ is also ordinary-recursive, if we take $T_0=\{x_1,x_3\}$
 and $T_1=\{x_1,x_2,x_3\}=D$. However,  $D$ is not simple-recursive.
\end{example}

\begin{example}\label{exa3}\rm
Consider a $5^{5-2}$ design $D$ with $x_4=x_1+x_2$ and $x_5=x_1+x_2+x_3$.
Take $T_0=\{x_1,x_2,x_3\}$, $T_1=\{x_1,x_2,x_3,x_4\}$ and $T_2=\{x_1,x_2,x_3,x_4,x_5\}=D$, then $D$ is simple-recursive.
If $x_5=x_1+x_2+2x_3$ instead, then $D$ is ordinary-recursive but not simple-recursive.
Consider another $5^{5-2}$ design $D$ with $x_4=x_1+x_2$ and $x_5=x_1+2x_2+2x_3$. This design is not recursive because $x_5$ is not involved in any word of length three. However, when one more column $x_6=x_1+2x_2$ is added, it is ordinary-recursive.
\end{example}

Regular designs with $q^2$ runs are commonly used in practice because they are economical and guarantee that linear terms are uncorrelated. Those designs accommodate two independent columns and up to $q-1$ dependent columns.
By Definition \ref{def1}, they are all recursive by letting $T_0$ include the two independent columns and $T_1=D$.
\begin{lemma}\label{lem1}
Let $q$ be an odd prime and $D$ be a regular design of $q^2$ runs. Then $D$ is recursive.
\end{lemma}

Clearly, recursive designs include ordinary-recursive designs, which in turn include simple-recursive designs.
For three-level designs, the three types of designs are equivalent, while for designs with more than three levels, they are dramatically different.
Table \ref{tab5} compares the numbers of the three types of designs with $25$ and $49$ runs.
The numbers of simple-recursive designs are much smaller than the numbers of the other two types of designs.
Although there is a difference between the numbers of ordinary-recursive and recursive designs, the difference is small.
As the number of columns increases, all designs tend to be ordinary-recursive.

\begin{table}[t!]
\caption{The numbers of the three types of recursive designs with 25 and 49 runs. \label{tab5}}
\vskip .2cm
\centerline{\tabcolsep=4truept\begin{tabular}{|ccccccccc|} \hline
   && \multicolumn{3}{c}{25-run designs}   & &  \multicolumn{3}{c|}{49-run designs}\\
$n$&&   simple & ordinary  & recursive &&   simple & ordinary  & recursive\\\hline
3 &&    2&  6 &  8 &&    2 & 10  & 18\\
4 &&   6 & 22 & 24 &&    6 & 99  & 135\\
5 &&  20 & 32 & 32 &&   20 & 517 & 540\\
6 &&  16 & 16 & 16 &&   70 & 1214& 1215\\
7 &&     &    &    &&  252 & 1458& 1458\\
8 &&     &    &    &&  267 & 729 & 729\\
\hline\end{tabular}}
\end{table}

The next theorem gives a sufficient condition for the $E_{b^*}$ to be the unique design with $\beta_3=0$ among all possible $q^m$ $E_b$'s.
\begin{theorem}\label{th2}
{For an odd prime $q$, let $D$ be a regular $q^{n-m}$ design defined by \eqref{eq31}, and $E_{b}$ be defined as \eqref{eqeb}.
If $D$ is ordinary-recursive}, then $E_{b^*}$ with $b^*$ defined in \eqref{eq1} is the only design with $\beta_3=0$ among all $q^m$ $E_b$'s derived from $D$.
\end{theorem}

\revb{In fact, we can show that if $D$ has no greater than 13 levels, the result of Theorem \ref{th2} can be extended beyond ordinary-recursive designs. That is,
we have the following more general result for $q\leq13$.

\begin{theorem}\label{thadd}
For a recursive $q^{n-m}$ design $D$, if $q$ is an odd prime and $q\leq13$, the $E_{b^*}$ with $b^*$ defined in \eqref{eq1} is the only design with $\beta_3=0$ among all $E_b$'s derived from $D$.
\end{theorem}
Theorem \ref{thadd} is not true for $q\geq17$. A counter example for $q=17$ comes with a $17^{3-1}$  design with $x_3=2x_1+4x_2$. By \eqref{eq1}, $b^*=14$. Then $E_{14}$ has $\beta_3=0$, while the design $E_4$ with columns $x_1,x_2$, and $x_3+4$ also has zero $\beta_3$. That being said, as the number of columns increases, the number of non-ordinary-recursive regular designs decreases dramatically so Theorem \ref{th2} works for most recursive designs with many columns. }


\begin{example}\label{exa7}\rm
Consider a $7^{8-6}$ design $D$ with $x_3=x_1+x_2, x_4=x_1+2x_2, x_5=x_1+4x_2, x_6=x_1+5x_2, x_7=2x_1+5x_2$, and $x_8=2x_1+6x_2$. There are $7^6=117,649$ $E_b$'s derived from $D$, which makes it cumbersome, if not impossible, to do an exhaustive search for the best $E_b$.  Note that $x_7=x_1+x_6$, $x_8=x_3+x_6$. So  $D$ is ordinary-recursive by taking $T_0=\{x_1,x_2\}$, $T_1=\{x_1,\ldots,x_6\}$ and $T_2=\{x_1,\ldots,x_8\}=D$. Equation  \eqref{eq1} gives $b_1^*=2, b_2^*=4, b_3^*=1, b_4^*=3, b_5^*=5,$ and $b_6^*=0$. It can be verified that  $\beta_3(E_{b^*})=0$ and $\beta_4(E_{b^*})=9.677$. By Theorem \ref{th2},  $E_{b^*}$ is the best design among all $E_b$'s derived from $D$ under the minimum $\beta$-aberration criterion.
\end{example}

\revb{By Theorems \ref{th2} and \ref{thadd},} for an ordinary-recursive design or a recursive design with no more than 13 levels, $E_{b^*}$ is the best design among all $E_b$'s, which is obtained without any computer search.
{Theorem \ref{th2} does not apply to the class of linearly permuted designs $D_b$'s. Here is a counter example.
\begin{example}\label{exa6}\rm
Consider the design $7^{3-1}$ design $D$ defined by $x_3=2x_1+2x_2$ in Example~\ref{exa4}. Example \ref{exa2} shows that it is ordinary-recursive,  
{but there are three $D_b$'s with zero $\beta_3$. Specifically, it is easy to verify that $\beta_3(D_b)=0$ for $b=0,3,5$. }
\end{example}
In fact, \cite{tang2014permuting} showed that if $D$ is simple-recursive, the design $D_{\tilde{b}}$ given by
\begin{equation}\label{eqdb}
\tilde{b}_i = \left(1-\sum_{j=1}^{n-m}c_{ij} \right) (q-1)/2   ~~~(i=1,\ldots,m).
\end{equation}
is the unique design with $\beta_3=0$ among all $D_b$'s.
As we have shown above, only a small amount of regular designs are simple-recursive.
Therefore, results on simple-recursive designs are usually not applicable for designs with more than three levels.
In contrast, Theorem \ref{th2} is more general and applies to the broader classes of ordinary-recursive and recursive designs.
}

\revb{Theorem \ref{thadd} and Lemma \ref{lem1} indicate} the following result.

\begin{corollary}\label{co1}
For an odd prime $q\leq13$, let $D$ be a regular design of $q^2$ runs. Then $E_{b^*}$ with $b^*$ defined as \eqref{eq1} is the unique design with $\beta_3=0$ among all $E_b$'s derived from $D$.
\end{corollary}

Now we show another useful property of $E_{b^*}$.
A design $D$ over $Z_{q}$ is called mirror-symmetric if $(q-1)J-D$ is the same design as $D$, where $J$ is a matrix of unity.  Mirror-symmetric designs  include two-level foldover designs as special cases.

\begin{theorem}\label{th4}
For an odd prime $q$, let $D$ be a regular $q^{n-m}$ design defined by \eqref{eq31}, and $E_{b}$ be defined as \eqref{eqeb}.
Then $E_{b^*}$ with $b^*$ defined in \eqref{eq1} is mirror-symmetric.
\end{theorem}

 {\cite{tang2014permuting} showed that} a design is mirror-symmetric if and only if it has $\beta_k=0$ for all odd $k$.
By Theorem \ref{th4}, the $E_{b^*}$ has $\beta_k(E_{b^*})=0$ for all odd $k$. This guarantees that {odd-order terms are not aliased with any even-order term. Specifically, linear terms are not aliased with any even-order term. }

\section{Construction Method and Design Comparisons}

Based on our theoretical results, we propose a sequential method for constructing multilevel nonregular designs. For simplicity, we focus on designs with $q^2$ runs although the method and results apply for general $q^{n-m}$ designs.
A regular fractional factorial design with $q^2$ runs has two independent columns, denoted as $x_1$ and $x_2$, and can accommodate up to $(q-1)$ dependent columns each of which is generated by $c_1x_1+c_2x_2$ with $c_1,c_2\in \{1,\ldots,q-1\}$.
Then the first two columns of $E_{b^*}$ are $W(x_1)$ and $W(x_2)$, respectively. To obtain $n\geq 3$ columns,
we add columns to $E_{b^*}$ sequentially by searching over generators $(c_1,c_2)$. Each new column is generated by
$W(c_1x_1+c_2x_2+b^*)$ where $b^*=(1-c_1-c_2)\gamma$ with $\gamma$ defined in \eqref{eqgamma} and $(c_1,c_2)$ minimizes $\beta_4(E_{b^*})$, that is,
$$
(c_1,c_2)=\arg\min_{(c_1,c_2)}\beta_4(E_{b^*}).
$$
The last three columns of Tables \ref{tab2}--\ref{tab121run} show the generators of the {added columns as well as the $\beta$-wordlength patterns of the obtained $E_{b^*}$.}

\begin{table}[t!]
\caption{Comparison of $\beta$-wordlength patterns for 25-run designs with 5 levels.}\label{tab2}
\vskip .2cm
\centerline{\tabcolsep=4truept\begin{tabular}{|lllllclllcll|} \hline
 && \multicolumn{2}{c}{$D$}  &&  \multicolumn{3}{c}{$D_{\tilde{b}}$}   & &  \multicolumn{3}{c|}{$E_{b^*}$}\\
$n$   &&  $\beta_3$ &  $\beta_4$  &&   Generators  &  $\beta_3$  & $\beta_4$ &&    Generators  &  $\beta_3$  & $\beta_4$ \\\hline
3   && 0.125 & 0.525   &&  (1,2)      &0~~  &0.271 && (1,1) &0~~ &0.027\\
4   && 0.375 & 1.361   &&  (2,1)      &0  &1.336 && (1,2) &0 &1.037\\
5   && 0.750 & 3.029   &&  (1,4)      &0  &3.793 && (1,3) &0 &3.768\\
6   && 1.250 & 6.786   &&  (1,1)      &0  &8.250 && (2,3) &0 &8.250\\
\hline
\end{tabular}}
\end{table}

To see the merit of $E_{b^*}$'s, we compare them with {commonly used regular designs and the class of $D_{\tilde{b}}$'s. The commonly used regular design \citep{mukerjee2006modern}, denoted by $D$,  consists of  the first $n$ columns of }
\begin{equation}\label{eq:q^2}
x_1, x_2, x_1+x_2,x_1+2x_2,x_1+3x_2, \ldots, x_1+(q-1)x_2.
\end{equation}
The design $D_{\tilde{b}}$ is obtained sequentially similar to the {generation} of $E_{b^*}$. The only difference is that the added column of $D_{\tilde{b}}$ is $c_1x_1+c_2x_2+\tilde{b}$ where $\tilde{b}=(1-c_1-c_2)(q-1)/2$.
Tables \ref{tab2}--\ref{tab121run} show the comparisons of such obtained designs $D$, $D_{\tilde{b}}$, and $E_{b^*}$ with 25 runs, 49 runs, and 121 runs, respectively.
We can see that the $E_{b^*}$ always performs the best for any design size.

\begin{table}[t!]
\caption{Comparison of $\beta$-wordlength patterns for 49-run designs with 7 levels.}\label{tab3}
\vskip .2cm
\centerline{\tabcolsep=4truept\begin{tabular}{|lllllclllcll|} \hline
 && \multicolumn{2}{c}{$D$}  &&  \multicolumn{3}{c}{$D_{\tilde{b}}$}   & &  \multicolumn{3}{c|}{$E_{b^*}$}\\
$n$   && $\beta_3$ &  $\beta_4$  &&   Generators &  $\beta_3$  & $\beta_4$ &&    Generators &  $\beta_3$  & $\beta_4$ \\\hline
3   && 0.063  &0.563   &&  (2,3)       &0  &0.063 && (1,1)     &0 &0.003\\
4   && 0.188  &1.354   &&  (1,4)       &0  &0.313 && (3,5)     &0 &0.055\\
5   && 0.375  &2.440   &&  (2,5)       &0  &1.135 && (3,6)     &0 &0.836\\
6   && 0.625  &4.313   &&  (1,2)       &0  &3.094 && (2,5)     &0 &2.368\\
7   && 0.938  &7.401   &&  (2,2)       &0  &6.438 && (2,6)     &0 &4.928\\
8   && 1.312  &12.78   &&  (2,6)   &0  &11.23 && (2,3)     &0 &9.677\\
\hline
\end{tabular}}
\end{table}

\begin{table}
	\caption{Comparison of $\beta$-wordlength patterns for 121-run designs {with 11 levels}. \label{tab121run}}
\vskip .2cm
\centerline{\tabcolsep=4truept\begin{tabular}{|lllllclllcll|}
\hline
&& \multicolumn{2}{c}{$D$}  &&  \multicolumn{3}{c}{$D_{\tilde{b}}$}   & &  \multicolumn{3}{c|}{$E_{b^*}$}\\
$n$   && $\beta_3$ &  $\beta_4$  &&   Generators &  $\beta_3$  & $\beta_4$ &&    Generators &  $\beta_3$  & $\beta_4$ \\\hline
			3   && 0.025  &0.585   &&  (2,4)       &0  &0.010 && (1,1)     &0 &0.0002\\
			4   && 0.075  &1.388   &&  (4,2)       &0  &0.055 && (2,4)     &0 &0.005\\
			5   && 0.150  &2.350   &&  (5,3)       &0  &0.281 && (4,2)     &0 &0.015\\
			6   && 0.250  &3.629   &&  (3,5)       &0  &0.710 && (2,9)     &0 &0.031\\
			7   && 0.375  &5.274   &&  (4,7)       &0  &1.466 && (2,8)     &0 &0.637\\
			8   && 0.525  &7.682  &&  (1,3)     &0  &3.152 && (5,3)     &0 &1.308\\
			9  && 0.700  &11.07   &&  (2,8)       &0  &5.519 && (4,10)     &0 &3.572\\
			10   && 0.900  &15.82   &&  (3,3)       &0  &8.891 && (1,7)     &0 &5.864\\
			11   && 1.125  &22.26   &&  (1,7)       &0  &13.49 && (5,1)     &0 &9.896\\
			12   && 1.375  &31.29   &&  (4,10)   &0  &19.65 && (5,4)     &0 &14.44\\
		\hline \end{tabular}}
\end{table}

\begin{figure}
  \centering
  \caption{Plot of $Mm_s$ (the larger the better) against $s$ for four designs: $D$ (circle), $D_{\tilde{b}}$ (cross), $E_{b^*}$ (square), the maximum-projection design (triangle), and the collapsed maximum-projection design (plus).}\label{Mms}
  \includegraphics[width=.6\textwidth]{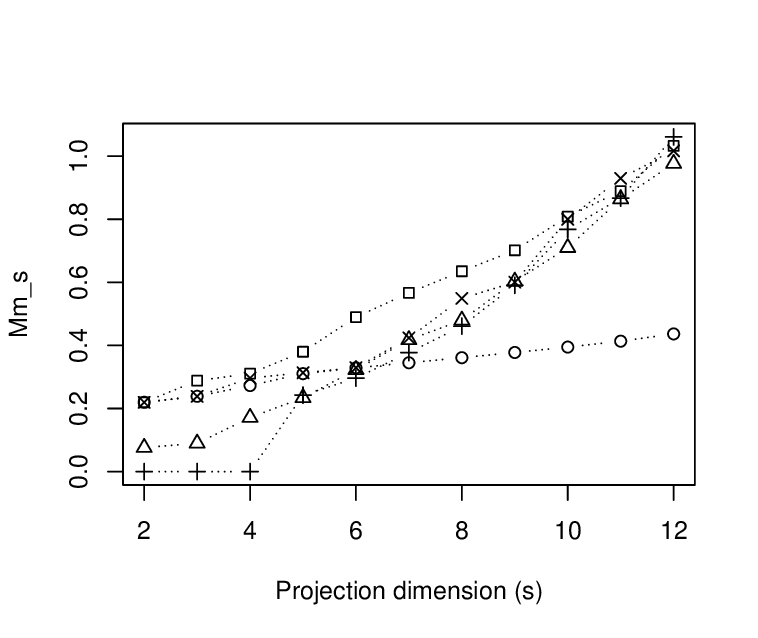}
\end{figure}

To illustrate the merit of the obtained designs $E_{b^*}$,
we further examine their space-filling property. For an $N\times n$ design, we consider the maximin measure in all projection dimensions, which is given by
$$
Mm_s=\min_{r=1,\ldots,\binom{n}{s}}\left\{ \frac{1}{\binom{N}{2}} \sum_{i=1}^{N-1}\sum_{j=i+1}^{N}\frac{1}{d_{ij,sr}^{2s}} \right\} ^{-1/(2s)},\mbox{ for }s=1,\ldots,n,
$$
where $d_{ij,sr}$ is the Euclidean distance between the $i$th and $j$th design points in the $r$th projection of
dimension $s$. Design points are scaled to $[0,1]^n$ to apply this measure, \revb{that is, the $j$th column is obtained by $x_j/(q-1)$.}
The measure was proposed in \cite{joseph2015maximum} to generate the so called ``maximin projection designs".
Designs with larger $Mm_s$ values are more space-filling in {$s$-dimension projections}. Figure \ref{Mms} plots the $Mm_s$ values of the $121\times 12$ designs in Table \ref{tab121run} for $s=1,\ldots, 12$. We also generate a $121\times 12$ maximum-projection LHD from R package MaxPro \citep{joseph2015maximum} and include its $Mm_s$ values in Figure \ref{Mms}. The design was claimed to be space-filling in all projected dimensions so may serve as a benchmark in the comparison. Because this design has 121 levels, we further collapse it to a 11-level design and include the $Mm_s$ values of the collapsed design in Figure \ref{Mms}.
To obtain a good maximum-projection design, the R package MaxPro is run 100 times and the best design is selected.
It takes on average 7 seconds to get a maximum-projection design. Therefore, to run the package 100 times takes about 12 minutes, whereas it takes less than a second to get \reva{any  of the other designs in the plot}. Even so,
Figure \ref{Mms} shows that $E_{b^*}$ outperforms the selected maximum-projection design and its collapsed design for all $s\leq11$ projection dimensions, although the collapsed design is marginally better than $E_{b^*}$ for the full dimension $s=12$.
Besides, $E_{b^*}$ performs better than all other designs in Figure \ref{Mms} on {projection dimension $s=2,\ldots,10$,} and is only slightly worse than $D_{\tilde{b}}$ when $s=11$. The good performance of $E_{b^*}$ comes from its zero $\beta_3$ and smaller $\beta_4$ values. 

We also {examine} such comparisons for designs of other sizes in Table \ref{tab121run} and {get} similar performance. This is because designs in Table \ref{tab121run} are obtained sequentially such that those with less than 12 columns are actually projections of the $121\times 12$ designs. Therefore, \reva{Figure~\ref{Mms} also reflects} the projection properties of designs with fewer columns.
Similar results also hold for 25-run and 49-run designs.

\section{Applications}
\label{sec:app}
Consider applying the three 25-run designs with 3 columns and 5 levels in Table \ref{tab2} {to the} following \revb{normalized} second-order polynomial model
\begin{equation}\label{eq25}
y=\alpha_{0}+\sum_{j=1}^{3}p_{1}(x_{j})\alpha_{j}+\sum_{j=1}^{3}p_{2}(x_{j})\alpha_{jj}+
\sum_{j=1}^{2}\sum_{k=j+1}^{3}p_{1}(x_{j})p_{1}(x_{k})\alpha_{jk}+\varepsilon,
\end{equation}
where $p_{1}(x)=\sqrt{2}(x-2)/2$, $p_{2}(x)=\sqrt{5/14}\{(x-2)^2-2\}$,
$\alpha_{0},\alpha_{j},\alpha_{jj}$, and $\alpha_{jk}$ are the intercept, linear, quadratic and bilinear terms, respectively, and $\varepsilon\sim N(0,\sigma^2)$.
\revb{Using such a normalized model instead of a model with natural terms (i.e., terms $x_j$, $x_j^2$, and $x_jx_k$) produces orthogonality between any two linear terms and between linear and quadratic terms for an orthogonal array.}
For the regular design $D$, because $\beta_3(D)\neq0$, linear terms are aliased or correlated with bilinear terms and the model in \eqref{eq25} is indeed not estimable.
While both  $D_{\tilde{b}}$ and $E_{b^*}$ have $\beta_{1}=\beta_{2}=\beta_{3}=0$, the intercept and all the linear terms are not correlated with the quadratic and bilinear terms and so they can be estimated independently.
For either design, let $M$ denote the model matrix corresponding to the 3 quadratic and 3 bilinear terms:
$\alpha_{11},\alpha_{22},\alpha_{33},\alpha_{12},\alpha_{13}$ and $\alpha_{23}$.
The variance-covariance matrix of the estimates of parameters for these terms is $\sigma^2(M^\mathrm{T}M)^{-1}$.
For $D_{\tilde{b}}$, the variances of the estimates for quadratic terms $\alpha_{11},\alpha_{22}$ and $\alpha_{33}$ are $0.047\sigma^2$, $0.041\sigma^2$, and $0.047\sigma^2$, respectively, and for bilinear terms $\alpha_{12},\alpha_{13}$ and $\alpha_{23}$ are $0.051\sigma^2$, $0.050\sigma^2$, and $0.051\sigma^2$, respectively.
For $E_{b^*}$, the variance of the estimate for each quadratic term is $0.040\sigma^2$, and for each bilinear term is $0.041\sigma^2$. With $E_{b^*}$, the variance of quadratic terms decreases by up to $14.9\%$ and the variance of bilinear terms decreases by up to $19.6\%$.
It can be verified that the correlations between the estimates are also smaller for $E_{b^*}$ than $D_{\tilde{b}}$.

Further, consider the bias brought by the inadequacy of polynomial terms {in model \eqref{eq25}}. Suppose there are nonnegligible third-order polynomial terms as
$$
\sum_{i+j+k=3} \alpha_{ijk} p_{i}(x_1) p_{j}(x_2) p_{k}(x_3). 
$$
Then the estimates of the linear parameters in model \eqref{eq25} are biased by these third-order terms.
Specifically,
{for the estimators from the design $D_{\tilde{b}}$, we have}
\begin{eqnarray*}
E(\hat{\alpha}_1) = \alpha_1-.12\alpha_{021}-.36\alpha_{012}+.3\alpha_{111}, \\
E(\hat{\alpha}_2) = \alpha_2+.36\alpha_{201}-.36\alpha_{102}-.1\alpha_{111}, \\
E(\hat{\alpha}_3) = \alpha_3+.36\alpha_{210}-.12\alpha_{120}-.3\alpha_{111},
\end{eqnarray*}
and {for the estimators from the design $E_{b^*}$, we have}
\begin{eqnarray*}
E(\hat{\alpha}_1) = \alpha_1+.096\alpha_{021}-.096\alpha_{012}+.08\alpha_{111}, \\
E(\hat{\alpha}_2) = \alpha_2+.096\alpha_{201}-.096\alpha_{102}+.08\alpha_{111}, \\
E(\hat{\alpha}_3) = \alpha_3+.096\alpha_{210}+.096\alpha_{120}-.08\alpha_{111}.
\end{eqnarray*}
Obviously, the design $E_{b^*}$ brings less bias to the estimators of linear terms than $D_{\tilde{b}}$.
Because $\beta_5=0$ for both designs, the estimates of second-order terms from $D_{\tilde{b}}$ and $E_{b^*}$ are not biased by third-order terms. \reva{In summary, $E_{b^*}$ is better for screening or studying quantitative factors than $D_{\tilde{b}}$ and $D_b$.  The results are general and apply to other designs in Tables  \ref{tab2}--\ref{tab121run}.}

\section{Concluding Remarks}

We provide a new class of nonregular designs via the Williams transformation and develop a theory on the property of the obtained designs.  Using the theory, we further propose a sequential method for constructing nonregular designs with minimum $\beta$-aberration.  The sequential method is fast and efficient to generate multilevel nonregular designs with large numbers of runs and factors. While two-level nonregular designs have been catalogued by some researchers, the
construction of multilevel nonregular designs was rarely studied. The approach in this paper is a pioneer work in this field. The obtained designs provide more accurate estimations on factorial effects and are more efficient than regular designs for screening quantitative factors.

The obtained designs can be used to generate orthogonal  LHDs which are commonly used and studied in \reva{computer experiments}.
Orthogonal LHDs have $\beta_1=\beta_2=0$ therefore guarantee the orthogonality between linear effects.
A popular construction, proposed by \cite{steinberg2006construction} and \cite{pang2009construction}, is to rotate a regular design to obtain an LHD which inherits the orthogonality from both the rotation matrix and the regular design. \cite{wang2018construction} improved the method by rotating a linearly permuted regular design, that is, the $D_{\tilde{b}}$ with $\tilde{b}$ defined in \eqref{eqdb}.
Such generated orthogonal LHDs have $\beta_3=0$ thus can guarantee that nonnegligible quadratic and bilinear effects do not contaminate the estimation of linear effects.
With the results in this paper, we can rotate the class of designs $E_{b^*}$ and obtain new orthogonal LHDs which have smaller $\beta_4$ values and inherit the good space-filling property of $E_{b^*}$. These LHDs may be good options for designing computer experiments and Gaussian processing modeling.

The Williams transformation is pairwise linear, which is probably the simplest nonlinear
transformation, yet it leads to some remarkable results such as Theorems \ref{th2} and \ref{th4}. It would
be of interest to identify and characterize other nonlinear transformations that have similar
properties. \revb{In addition, the proposed method requires the number of levels of regular designs being a prime number and does not work for, say, four-level designs. It would also be interesting to extend the method for non-prime numbers of levels.}


\vskip 14pt
\noindent {\large\bf Acknowledgements}

The authors thank two reviewers for their helpful comments.
\par

\section*{Appendix: Proofs}

We need the following lemmas for the proofs.
\begin{lemma}\label{lem3}
The $D_b$ is the same design as $D_e+\gamma \modq$, where $e=b-b^*$, $\gamma$ is defined as \eqref{eqgamma}, and $b^*$ is defined as \eqref{eq1}.
\end{lemma}
\begin{proof}
For $D_b$, permuting all columns $x_j$ to $x_j-\gamma$ for $j=1,\ldots, n$ is equivalent to keeping the independent columns unchanged while permuting the dependent columns $x_{n-m+i}+b_i$ to $x_{n-m+i}+b_i-b_i^*$ for $i=1,\ldots,m$.
Hence, $D_b-\gamma$ is the same design as $D_e$ with $e=b-b^*$. Equivalently, $D_b$ is the same design as $D_e+\gamma\modq$.
\end{proof}

\begin{lemma}\label{lem5}
If x is a real number which is not an integer, then
$$
\sum_{n=-\infty}^{\infty}\frac{(-1)^{n-1}}{(n+x)^2}=\frac{\pi^2\cos\pi x}{(\sin\pi x)^2}.
$$
\end{lemma}
\begin{proof}
It is known that $\sum_{n=-\infty}^{\infty}1/(n+x)^2=\pi^2/(\sin\pi x)^2.$ Then
\begin{eqnarray*}
  \sum_{n=-\infty}^{\infty}\frac{(-1)^{n-1}}{(n+x)^2}&=&\sum_{n=-\infty}^{\infty}\frac{1}{(n+x)^2}-2\sum_{{\rm even}~n}\frac{1}{(n+x)^2}\\
&=&\frac{\pi^2}{(\sin\pi x)^2}-\frac{1}{2}\frac{\pi^2}{(\sin(\pi x/2))^2} \\
   &=& \frac{\pi^2\cos\pi x}{(\sin\pi x)^2}.
\end{eqnarray*}
\end{proof}

\begin{lemma}\label{lem6}
Let $p_1(x)=\rho[x-(q-1)/2]$ be the linear orthogonal polynomial, where $\rho=\sqrt{12/[(q+1)(q-1)]}$.
Then for $x=0,\ldots,q-1$,
$$
p_1(x) =-\frac{\rho}{2q}\sum_{v=0}^{q-1}g(v)\cos\left\{\frac{(2v+1)\pi (x+0.5)}{q}\right\}.
$$
where
\begin{equation}\label{eqg}
  g(v)=\frac{\cos(\pi (v+0.5)/q)}{ \{\sin(\pi (v+0.5)/q)\}^2}.
\end{equation}
\end{lemma}

\begin{proof}
For $x\in(0,q)$, the Fourier-cosine expansion of $x-q/2$ is given by
$$
x-\frac{q}{2} = \sum_{v=1}^{\infty} a_v \cos\left(\frac{v\pi x}{q}\right),
$$
with
$$
a_v = \frac{2}{q}\int_{0}^{q}\left(x-\frac{q}{2}\right)\cos\left(\frac{v\pi x}{q}\right)dx
=\left\{
   \begin{array}{ll}
     0, & \hbox{if $v$ is even;} \\
     -4q/(v^2\pi^2), & \hbox{if $v$ is odd.}
   \end{array}
 \right.
$$
Then
\begin{eqnarray*}
p_1(x) &=& -\frac{4\rho q}{\pi^2} \sum_{{\rm odd}~v>0} \frac{1}{v^2} \cos\left(\frac{v\pi (x+0.5)}{q}\right)
\\ &=& -\frac{2\rho q}{\pi^2} \sum_{v=-\infty}^{\infty}  \frac{1}{(2v+1)^2} \cos\left\{\frac{(2v+1)\pi (x+0.5)}{q}\right\}
\\ &=& -\frac{2\rho q}{\pi^2} \sum_{k=-\infty}^{\infty} \sum_{v=0}^{q-1} \frac{1}{(2kq+2v+1)^2} \cos\left\{\frac{(2kq+2v+1)\pi (x+0.5)}{q}\right\}.
\end{eqnarray*}
Since for any integers $k$ and $x$,
$$
\cos\left\{\frac{(2kq+2v+1)\pi (x+0.5)}{q}\right\}= (-1)^k \cos\left\{\frac{(2v+1)\pi (x+0.5)}{q}\right\}, $$
we have
$$
p_1(x) =  -\frac{2\rho q}{\pi^2}\sum_{v=0}^{q-1} \sum_{k=-\infty}^{\infty} \frac{(-1)^{k}}{(2kq+2v+1)^2}  \cos\left\{\frac{(2v+1)\pi (x+0.5)}{q}\right\}.
$$
By Lemma \ref{lem5} and \eqref{eqg},  
we have
$$
p_1(x) =-\frac{\rho}{2q}\sum_{v=0}^{q-1}g(v)\cos\left\{\frac{(2v+1)\pi (x+0.5)}{q}\right\}.
$$
\end{proof}

\begin{proof}[{Proof of Theorem \ref{th1}}]
Denote $e=b-b^*$ and $D_e=(y_{ij})$.
By Lemma \ref{lem3}, $D_b$ is the same design as $(D_e+\gamma)\modq$, so $E_b=W(D_b)=W(D_e+\gamma)$.
By Lemma \ref{lem6},
\begin{eqnarray*}
p_1\left(W(x)\right) &=& -\frac{\rho}{2q}\sum_{v=0}^{q-1}g(v)\cos\left\{\frac{(2v+1)\pi (W(x)+0.5)}{q}\right\}
\\ &=& -\frac{\rho}{2q}\sum_{v=0}^{q-1}g(v)\cos\left\{\frac{(2v+1)\pi (2x+0.5)}{q}\right\}
\end{eqnarray*}
because
$ \cos\left\{{(2v+1)\pi (W(x)+0.5)}/{q}\right\} = \cos\left\{{(2v+1)\pi (2x+0.5)}/{q}\right\}$ for any integer $v$.
Then we have
\begin{eqnarray}\label{eq4}
\beta_3(E_b)&=&\beta_3(W(D_e+\gamma))\nonumber\\
&=&N^{-2}\sum_{y_1,y_2,y_3}\left|\sum_{i=1}^{N}p_1(W(y_{i1}+\gamma))p_1(W(y_{i2}+\gamma))p_1(W(y_{i3}+\gamma))\right|^2\nonumber\\
&=&N^{-2}\left(\frac{\rho}{2q}\right)^6\sum_{y_1,y_2,y_3}\left| \sum_{v_1=0}^{q-1}\sum_{v_2=0}^{q-1}\sum_{v_3=0}^{q-1}g(v_1)g(v_2)g(v_3)S(y,v) \right|^2,
\end{eqnarray}
where $\sum_{y_1,y_2,y_3}$ sums over all three different columns $y_1,y_2,y_3$ in $D_e$, $y_j=(y_{1j},\ldots,y_{Nj})$ for $j=1,2,3$,
and
\begin{eqnarray*}
  S(y,v) &=& \sum_{i=1}^{N}\prod_{j=1}^{3} \cos\left\{\frac{(2v_j+1)\pi (2y_{ij}+2\gamma+0.5)}{q}\right\} \\
    &=& \sum_{i=1}^{N}\prod_{j=1}^{3} {(-1)^{(q+1)/2+v_j}} \sin\left\{\frac{2(2v_j+1)\pi y_{ij}}{q}\right\}\\
    &=& {(-1)^{(q+1)/2+v_1+v_2+v_3}} \sum_{i=1}^{N}\prod_{j=1}^{3} \sin\left\{\frac{2(2v_j+1)\pi y_{ij}}{q}\right\}.
\end{eqnarray*}
If $b=b^*$, $e=0$ and $D_e=D$. Because $D$ is a regular design, it is a linear space over $Z_q$. Thus, $(q-y_{i1},\ldots,q-y_{in})\in D$ whenever $(y_{i1},\ldots,y_{in})\in D$. Then $S(y,v)=0$ for any $y=(y_1,y_2,y_3)$ and $v=(v_1,v_2,v_3)$. By \eqref{eq4}, $\beta_3(E_{b^*})=0$.
\end{proof}

\begin{proof}[{Proof of Theorem \ref{th2}}]
Following the proof of Theorem 1, if $b\neq b^*$, then $e=b-b^*$ has nonzero components.
{Since $D$ is ordinary-recursive, there exist three columns, say $z_1, z_2, z_3$, in $D$ such that  $z_3=c_1z_1+c_2z_2$,  $c_1=1$ or $-1$, $c_2\in Z_q$, and $z_1$, $z_2$ and $z_3+e_0$  are three columns in $D_e$, where  $e_0$ is a nonzero component of $e$}.
We only consider $c_1=1$ below as the proof for $c_1=-1$ is similar.
{
Let $d$ be the design formed by $z_1$, $z_2$, and $z_3+e_0$}.
By \eqref{eq4}, we only need to show that $\beta_3(W(d))\neq0$.
Note that
\begin{equation}\label{eq8}
\beta_3(W(d))=N^{-2}\left(\frac{\rho}{2q}\right)^6\left| \sum_{v_1=0}^{q-1}\sum_{v_2=0}^{q-1}\sum_{v_3=0}^{q-1} (-1)^{v_1+v_2+v_3} g(v_1)g(v_2)g(v_3)S(z,v) \right|^2,
\end{equation}
where $g(v)$ is defined in \eqref{eqg}, and
$$
S(z,v)=\sum_{i=1}^{N}\sin\left(\frac{2(2v_1+1)\pi z_{i1}}{q}\right)\sin\left(\frac{2(2v_2+1)\pi z_{i2}}{q}\right)\sin\left(\frac{2(2v_3+1)\pi (z_{i3}+e_0)}{q}\right).
$$
By applying the product-to-sum identities twice, we have
\begin{eqnarray}\label{eq3}
S(z,v)
&&=\frac{1}{4}\Bigg\{\sum_{i=1}^{N}\sin\left(\frac{2\pi (t_1z_{i1}-t_4z_{i2}+(2v_3+1)e_0)}{q}\right) \nonumber\\
&&+\sum_{i=1}^{N}\sin\left(\frac{2\pi (t_2z_{i1}+t_4z_{i2}-(2v_3+1)e_0)}{q}\right) \nonumber\\
&&-\sum_{i=1}^{N}\sin\left(\frac{2\pi (t_1z_{i1}+t_3z_{i2}+(2v_3+1)e_0)}{q}\right) \nonumber\\
&&-\sum_{i=1}^{N}\sin\left(\frac{2\pi (t_2z_{i1}-t_3z_{i2}-(2v_3+1)e_0)}{q}\right)\Bigg\},
\end{eqnarray}
where $t_1=2(v_1+v_3)+2$, $t_2=2(v_1-v_3)$, $t_3=2(v_2+v_3c_2)+c_2+1$, and $t_4=2(v_2-v_3c_2)-c_2+1$.
Let  
\begin{equation}\label{eq:v20}
v_{10} = q-1-v_3 \mbox{ and } v_{20} = v_3c_2+(c_2-1)(q+1)/2\modq.
\end{equation}
When $v_1=v_{10}$ and $v_2=v_{20}$, $t_1=t_4=0\modq$ and the first item in the right hand side of \eqref{eq3}, $\sum_{i=1}^{N}\sin\left(2\pi (t_1z_{i1}-t_4z_{i2}+(2v_3+1) e_0)/q\right)$, equals $N\sin(2\pi (2v_3+1)e_0/q)$. When $v_1\neq v_{10}$ or $v_2\neq v_{20}$, the item is zero.
By similar analysis to other items in \eqref{eq3}, we have
$$
S(z,v)=\left\{
         \begin{array}{ll}
           \frac{N}{4} \sin\left\{\frac{2\pi (2v_3+1)e_0}{q}\right\}, & \hbox{if $(v_1,v_2)=(v_{10}, v_{20})$ or $(q-1-v_{10}, q-1-v_{20})$;} \\
           - \frac{N}{4} \sin\left\{\frac{2\pi (2v_3+1)e_0}{q}\right\}, & \hbox{if $(v_1,v_2)=(v_{10},q-1-v_{20})$ or $(q-1-v_{10},v_{20})$;} \\
           0, & \hbox{otherwise.}
         \end{array}
       \right.
$$
Note that $g(q-1-v)=-g(v)$ for any $v$. Then by \eqref{eq8},
\begin{equation}\label{eq5v1}
\beta_3(W(d))
=\left(\frac{\rho}{2q}\right)^6\left| \sum_{v_3=0}^{q-1}(-1)^{v_3c_2}g(v_{20})(g(v_3))^2\sin \left\{\frac{2\pi (2v_3+1)e_0}{q}\right\}  \right|^2,
\end{equation}
where $v_{20}$ is defined in \eqref{eq:v20}. Applying $g(q-1-v)=-g(v)$ again, we can simply \eqref{eq5v1} as
\begin{equation}\label{eq5}
\beta_3(W(d))
=\frac{\rho^6}{16q^6} \left| \sum_{v_3=0}^{(q-1)/2}(-1)^{v_3c_2}g(v_{20})(g(v_3))^2\sin \left\{\frac{2\pi (2v_3+1)e_0}{q}\right\}  \right|^2.
\end{equation}
By considering the Taylor expansion of $g(v)$, we can see that the sum in \eqref{eq5} is dominated by the first two items with $v_3=0$ and $v_3=1$. It can be verified that \eqref{eq5} is nonzero for $e_0=1,\ldots,q-1$.
This completes the proof.
\end{proof}

\begin{proof}[{Proof of Theorem \ref{thadd}}]
Following the same process as in the proof of Theorem \ref{th2}, if $D$ is recursive, then for the three columns $z_1,z_2,$ and $z_3$ in $D$, $z_3=c_1z_1+c_2z_2$, where both $c_1$ and $c_2$ can be any value in $Z_q$. Then we can get \eqref{eq3} with $t_1$ and $t_2$ replaced by $t_1'=2(v_1+v_3c_1)+1+c_1$ and $t_2'=2(v_1-v_3)+1-c_1$, which will in turn result in a change of $v_{10}$ in \eqref{eq:v20} to
$$
v_{10}'=\left\{
          \begin{array}{ll}
            (q-1)/2-c_1/2-v_3c_1 \modq, & \hbox{if $c_1$ is an even number;} \\
            q-(c_1+1)/2-v_3c_1 \modq, & \hbox{if $c_1$ is an odd number.}
          \end{array}
        \right.
$$
Similar to \eqref{eq5}, we have
\begin{equation}\label{eqadd}
\beta_3(W(d))
=\frac{\rho^6}{16q^6} \left| \sum_{v_3=0}^{(q-1)/2}(-1)^{v_3c_2}g(v_{10}')g(v_{20})(g(v_3))\sin \left\{\frac{2\pi (2v_3+1)e_0}{q}\right\}  \right|^2.
\end{equation}
It can be verified that, for $q\leq13$, \eqref{eqadd} is nonzero for $e_0=1,\ldots,q-1$ for any $c_1,c_2\in Z_q$.
This completes the proof.
\end{proof}

\begin{proof}[{Proof of Theorem \ref{th4}}]
We need to show that for any run $W(x_1,\ldots,x_n)$ in $E_{b^*}$, $(q-1)-W(x_1,\ldots,x_n)$ also belongs to $E_{b^*}$. This is equivalent to show that for each run $(x_1,\ldots,x_n)$ in $D_{b^*}$, $W^{-1}(q-1-W(x_1,\ldots,x_n))$ also belongs to $D_{b^*}$.
Since the design $D$ contains the zero point $(0,\ldots,0)$, by Lemma \ref{lem3}, $D_{b^*}$ contains the point $(\gamma,\ldots,\gamma)$.
Because all design points of $D$ form a linear space and $D_b$ is a coset of $D$, then $\gamma-(x_1,\ldots,x_n)$ belongs to the null space of $D_{b^*}$.
Hence, $\gamma-(x_1,\ldots,x_n)+\gamma=2\gamma-(x_1,\ldots,x_n)$ belongs to $D_{b^*}$.
For $x=0,\ldots,q-1$,
$$
W^{-1}(x)=\left\{
            \begin{array}{ll}
              x/2, & \hbox{for even $x$;} \\
              q-(x+1)/2, & \hbox{for odd $x$,}
            \end{array}
          \right.
$$
and
\begin{eqnarray*}
W^{-1}(q-1-x)&=&\left\{
            \begin{array}{ll}
              (q-1)/2-W^{-1}(x), & \hbox{for even $x$;} \\
              (3q-1)/2-W^{-1}(x), & \hbox{for odd $x$,}
            \end{array}
          \right.\\
&=&2\gamma-W^{-1}(x).
\end{eqnarray*}
Then $W^{-1}(q-1-W(x_1,\ldots,x_n))=2\gamma-(x_1,\ldots,x_n)$. Hence, $W^{-1}(q-1-W(x_1,\ldots,x_n))$ belongs to $D_{b^*}$. This completes the proof.
\end{proof}

%

\markboth{\hfill{\footnotesize\rm LIN WANG AND HONGQUAN} \hfill}
{\hfill {\footnotesize\rm A CLASS OF MULTILEVEL NONREGULAR DESIGNS} \hfill}

\bibhang=1.7pc
\bibsep=2pt
\fontsize{9}{14pt plus.8pt minus .6pt}\selectfont
\renewcommand\bibname{\large \bf References}
\expandafter\ifx\csname
natexlab\endcsname\relax\def\natexlab#1{#1}\fi
\expandafter\ifx\csname url\endcsname\relax
  \def\url#1{\texttt{#1}}\fi
\expandafter\ifx\csname urlprefix\endcsname\relax\def\urlprefix{URL}\fi

\lhead[\footnotesize\thepage\fancyplain{}\leftmark]{}\rhead[]{\fancyplain{}\rightmark\footnotesize{} }
\bibliographystyle{asa}
\bibliography{ref}

\vskip .65cm
\noindent
Department of Statistics, The George Washington University, Washington DC 20052, USA.
\vskip 2pt
\noindent
E-mail: linwang@gwu.edu
\vskip 2pt

\noindent
Department of Statistics, University of California, Los Angeles, California 90095, USA.
\vskip 2pt
\noindent
E-mail: hqxu@stat.ucla.edu
\end{document}